\documentclass[pra,aps,10pt,superscriptaddress,a4paper,twocolumn]{revtex4}
\usepackage[mathletters]{ucs}
\usepackage[T1]{fontenc}
\usepackage[utf8x]{inputenc}
\usepackage{amsmath,amsthm,amscd}
\usepackage{latexsym}
\usepackage{graphicx}
\usepackage{booktabs}
\usepackage{hyperref}

\newtheorem{definition}{Definition}
\newtheorem{proposition}{Proposition}
\newtheorem{corollary}{Corollary}
\newtheorem{observation}{Observation}
\newtheorem{remark}{Remark}

\newcommand{\<}{\langle}
\renewcommand{\>}{\rangle}
\newcommand{\C}{{\cal C}}

\DeclareUnicodeCharacter{981}{\phi}

\begin{document}

\title{Low-dimensional quite noisy bound entanglement with
  cryptographic key}

\author{\L{}ukasz Pankowski}
\affiliation{Institute of Informatics, University of Gda\'nsk,
  Gda\'nsk, Poland}
\affiliation{Institute of Theoretical Physics and Astrophysics,
  University of Gda\'nsk, Gda\'nsk, Poland}

\author{Micha{\l} Horodecki}
\affiliation{Institute of Theoretical Physics and Astrophysics,
  University of Gda\'nsk, Gda\'nsk, Poland}

\begin{abstract}
  We provide a class of bound entangled states that have positive
  distillable secure key rate.  The smallest state of this kind is $4
  \otimes 4$.
  Our class is a generalization of the class presented in
  \cite{lowdim-pptkey}.  It is much wider, containing, in particular,
  states from the boundary of PPT entangled states (all of the states
  in the class in \cite{lowdim-pptkey} were of this kind) but also
  states inside the set of PPT entangled states, even, approaching the
  separable states.
  This generalization comes with a price: for the wider class a
  positive key rate requires, in general, apart from the \emph{one-way}
  Devetak-Winter protocol (used in \cite{lowdim-pptkey}) also
  the recurrence preprocessing and thus effectively is a \emph{two-way}
  protocol.
  We also analyze the amount of noise that can be admixtured to the
  states of our class without losing key distillability property
  which may be crucial for experimental realization.
  The wider class contains key-distillable states with higher entropy
  (up to 3.524, as opposed to 2.564 for the class in
  \cite{lowdim-pptkey}).
\end{abstract}

\maketitle

\section{Introduction}

Quantum cryptography, pioneered by Wiesner \cite{Wiesner}, allows to
obtain cryptographic key based on physical impossibility of
eavesdropping. Namely, if the transmitted signal is encoded into
quantum states, then by reading it, eavesdropper always introduces
noise into the signal. Thus Alice and Bob -- the parties who want to
communicate privately -- can measure the level of noise and detect
whether their transmission is secure (even if the noise was solely due
to eavesdropping). There are two types of quantum key distribution
protocols: \emph{prepare and measure} (as the original BB84 protocol
\cite{BB84}) and protocols based on a shared entangled state (originated
from the Ekert's protocol \cite{Ekert91}). For quite a time security
proofs of prepare and measure protocols had been based on showing
equivalence to the distillation (by local operations and classical
communication) of maximally entangled states (the first such proof is
due to Shor and Preskill \cite{ShorPreskill}).  It have led to a
belief that security of the quantum cryptography is always connected
to the distillation of the maximally entangled states (this issue was
perhaps first touched by Gisin and Wolf \cite{GisinWolf_linking}).

This belief suggested that one could not obtain secure key from bound
entangled states \cite{bound}, i.e., states from which maximally
entangled states cannot be distilled. On the contrary, the
key-distillable bound entangled states have been found \cite{pptkey}
and examples of low dimensional states have been provided
\cite{lowdim-pptkey}.  The multipartite case was also considered
\cite{AugusiakH2008-multi}.  There are two approaches to obtaining
cryptographic key from bound entangled PPT states: one is based on
approximating private bit with a PPT state \cite{pptkey,keyhuge} and
the other one -- on mixing orthogonal private bits
\cite{lowdim-pptkey}.

This paper continues on the second approach. The low dimensional
key-distillable states with positive partial transpose \footnote{If a
  state has positive partial transpose (PPT) then one cannot distill
  maximally entangled state from it. It is an long-standing open
  question whether PPT is also a necessary condition for
  non-distillability of maximal entanglement \cite{RMPK-quant-ent}
  (for recent development, see \cite{PankowskiPHH-npt2007}).} (hence,
bound entangled) presented in \cite{lowdim-pptkey} were lying on the
boundary of PPT states and existence of the key-distillable states
inside of PPT states was argued by the continuity argument, without
giving the explicit form of those inner states. In this paper we
present a wider class of PPT entangled key-distillable states
including states inside the set of PPT states even approaching the set
of separable states. We analyze properties of this class, as well as
provide some more general criteria of key distillability, by
exploiting criterion provided in \cite{acin-2006-73}.  This criterion
was earlier applied to analyze some PPT states in \cite{Bae2008} (see
also \cite{PhysRevA.75.032306} in this context).

The motivation behind the search for new bound entangled states with
distillable key, is two-fold. First of all, there is a fundamental
open question, whether from all entangled states one can draw secure
key.  To approach this question, one needs, in particular, to gather
more phenomenology on the issue of drawing key from bound entangled
states. In this paper, we have pushed this question a bit by showing
explicitly
that PPT key-distillable states can be in the interior of PPT states,
even, approaching the set of separable states.  Also, our general
criterion of key distillability can serve for searching to what extent
entanglement can provide for secure key.

Another motivation comes from recent experiments, where bound
entanglement was implemented in labs
\cite{Experimental-bound-Bourennane,Experimental-bound-Lavoie,2010PhRvA..81d0304K,2010arXiv1005.1965B,2010arXiv1006.4651D}.
In the experiments, usually, a four-partite bound entangled Smolin
state was used, which allows for a number of non-classical effects
being manifestations of true entanglement content of such state.  We
believe that low-dimensional bound entangled key-distillable states
are also good candidates for experimental implementation, providing a
non-classical effect -- possibility of distilling secure key.  This
requires states which are robust against noise, to facilitate the
process of preparing them in a lab. In this paper, we analyze
robustness of key-distillable states as well as provide very noisy
states, having, in particular, relatively large entropy (\mbox{c.a.}
3.5 bits versus 4 bits of maximal possible entropy). Last but not
least, the key-distillable bound entangled states are strictly related
to the effect of superactivation of quantum capacity
\cite{SmithYard-2008}, and our class may be further analyzed in this
respect (in this paper, we have provided some exemplary calculations).

The paper is organized as follows. In \mbox{Sec.} \ref{sec:prelim} we
review basic facts about general theory of distillation of secure key
from quantum states of \cite{keyhuge}. In particular, we describe
technique called the privacy squeezing. In \mbox{Sec.} \ref{sec:class} we
introduce our class of states which are PPT and key-distillable. We
verify that they lie inside the set of PPT states, touching the set of
separable states. Moreover, we check robustness of the property of
key-distillability. We also give the explicit form of an important
subset of our states as
mixtures of pure states in \mbox{Sec.} \ref{sec:exp}). In
\mbox{Secs.} \ref{sec:entropy} and \ref{sec:erasure} we examine entropic
properties of our states and their relation with Smith-Yard
superactivation of quantum capacity phenomenon. Finally, in
\mbox{Sec.} \ref{sec:general-key}, we provide a general sufficient
condition for distilling private key from quantum states of local
dimension not less than 4.

\section{Preliminaries}
\label{sec:prelim}

Let us first recall some important concepts of classical key
distillation from quantum states, covered in detail in \cite{keyhuge}.

A general state containing at least one bit of perfectly secure key is
called the \emph{private bit} or \emph{pbit} \cite{keyhuge}.  A
private bit in its so-called $X$-form is given by
\begin{align}
  \label{eq:pbit}
  \gamma(X) = \frac12
  \begin{bmatrix}
    \sqrt{XX^\dagger} & 0 & 0 & X \\
    0 & 0 & 0 & 0 \\
    0 & 0 & 0 & 0 \\
    X^\dagger & 0 & 0 & \sqrt{X^\dagger X}
  \end{bmatrix}
\end{align}
where $X$ is an arbitrary operator satisfying $\|X\|=1$ (here and
throughout the paper, we use the trace norm, that is the sum of the
singular values of an operator). The private bit has four subsystems:
$ABA'B'$ where block matrix \eqref{eq:pbit} represents $AB$ subsystem
and the blocks are operators acting on an $A'B'$ subsystem. Subsystems
$A$ and $B$ are single qubit subsystems while dimensions of $A'$ and
$B'$ must be greater or equal to 2, we assume dimensions $A'$ and $B'$
are equal and denote them by $d$.  Subsystem $AA'$ belongs to Alice
while subsystem $BB'$ belongs to Bob.  Every state presented
in the block matrix form throughout the paper has this structure.  The
bit of key contained in a private bit is
obtained by measuring subsystems $A$ and $B$ in the standard basis;
therefore, subsystem $AB$ is called the \emph{key part} of the state,
while subsystem $A'B'$ is called the \emph{shield} of the state, as it
protects correlations contained in the key part from an
eavesdropper.  Note that it may happen that Eve possesses a copy of
the shield subsystem (when, e.g., the shield consists of two flag states
– states with disjoint support) yet it does not hurt because the very
presence of the shield subsystem in Alice and Bob's hands
protects the bit of key.

For a general state with $ABA'B'$ subsystems (i.e., not necessarily a
private bit) one can infer possibility of distillation of private key
using the method called the \emph{privacy squeezing} \cite{keyhuge}.
Namely, we consider the following type of protocols: one measures the
key part in the standard basis and classically process the outcomes
(\mbox{cf.} \cite{acin-2006-73} for two-qubit states).  Given a
protocol of this type we would like to know whether it can distill key
from the state.  To this end, we construct a two qubit state in the
following way: one applies to the original state the so-called
\emph{twisting} operation, i.e., a unitary transformation of the
following form
\begin{align}
  U = \sum_{ij} |ij⟩_{AB}⟨ij| ⊗ U_{ij}^{A'B'}
\end{align}
and perform partial trace over $A'B'$.  Now, it turns out that if we
apply the protocol to the original state we obtain no less key than we
would obtain from the above two qubit state using the same protocol.

Therefore, if we apply a cleverly chosen twisting, we may infer
key-distillability of the original state, by examining a two-qubit
state (i.e., a much simpler object). This technique is called
the \emph{privacy squeezing}. The role of twisting is to `squeeze' the
privacy present in the original state into its key part, where it is
then more easily detectable, e.g., by protocols designed for two-qubit
states (see e.g., \cite{Gottesman-Lo,acin-2006-73,Renner2006-PhD}).

To explain why the two qubit state cannot give more key than the
original state (within the considered class of protocols) we invoke
the following result of \cite{keyhuge}.  One considers a state of
three systems: a quantum one -- Eve's system and two classical ones --
the registers holding the outcomes of measurement of the key part (the
state is therefore called a \emph{ccq state}). Now, it turns out that
twisting does not change this state.  However, in the considered class
of protocols Alice and Bob use only classical registers, so the output
of such protocols depends solely on the ccq state. Thus the key
obtained with and without twisting is exactly the same. This holds,
even though twisting is a non-local operation and the resulting state
can be more powerful under all other respects (such as drawing key by
some other type of protocols). Next, if we additionally trace out the
shield, i.e., the subsystem $A'B'$, this means that the resulting ccq
state differs from the original ccq state only by Eve having, in
addition, the shield. Thus, if any key can be obtained from it, it can
only be less secure than the key obtained from the original ccq state.

It turns out that for any `spider' state, i.e., state of the form
\begin{align}
  ρ = \begin{bmatrix}
    C & & & D \\
    & E & F \\
    & F^{†} & E' \\
    D^{†} & & & C'
  \end{bmatrix}
  \label{eq:spider}
\end{align}
\begin{figure*}[t]
  \centering
  \begin{align}
    \label{eq:the-class-blocks}
    \varrho = \frac12
    \begin{bmatrix}
      (λ_1 + λ_2) \sqrt{XX^\dagger} &  & & (λ_1 - λ_2) X \\
      \, & (λ_3 + λ_4)\sqrt{YY^\dagger} & (λ_3 - λ_4) Y \\
      \, & (λ_3 - λ_4)Y^\dagger & (λ_3 + λ_4)\sqrt{Y^\dagger Y} \\
      (λ_1 - λ_2) X^\dagger & & & (λ_1 + λ_2) \sqrt{X^\dagger X}
    \end{bmatrix}
  \end{align}
  \caption{Block matrix form of mixture of four private bits.}
  \label{fig:the-class-blocks}
\end{figure*}%
(where we have omitted zero blocks for clarity) there exists such a
twisting operation that the matrix elements of the two qubit state,
obtained by tracing out the $A'B'$ subsystem after applying the
twisting, are equal to trace norms of the corresponding blocks in the
original state:
\begin{align}
  σ = \begin{bmatrix}
    \|C\| & & & \|D\| \\
    & \|E\| & \|F\| \\
    & \|F\| & \|E'\| \\
    \|D\| & & & \|C'\|
  \end{bmatrix}
  \label{eq:ps-spider}
\end{align}
(we use here that $\|A\|=\|A^\dagger\|$ for trace norm).  This
twisting is in a sense optimal for the spider states. We call the two
qubit state \eqref{eq:ps-spider} the \emph{privacy-squeezed state} of
the original state. If a spider state satisfies $\|C\|=\|C'\|$ and
$\|E\|=\|E'\|$ than its privacy-squeezed state is a Bell diagonal
state.

For a deeper discussion of the privacy squeezing see \cite{keyhuge},
although the name \emph{spider state} is not used there.

\section{Distilling key from PPT mixtures of private states}
\label{sec:class}

Here, we construct a class of bound entangled states which are
key-distillable. They are mixtures of four orthogonal private bits of
some special form. We provide a sufficient condition to distill
cryptographic key from our class. The condition given in this section
is generalized to an arbitrary state in \mbox{Sec.}
\ref{sec:general-key}.

\subsection{Definition of the class}

Let us consider a class of states
\begin{align}
  \label{eq:the-class}
  \varrho = λ_1 γ_1^+ + λ_2 γ_1^- + λ_3 γ_2^+ + λ_4 γ_2^-
\end{align}
which is a mixture of four orthogonal private bits which could be
considered analogues to the Bell states.  The construction is possible
in dimension $2d \otimes 2d$, with $d\geq 2$.

The four private bits are given by
\begin{align}
  \label{eq:the-class-gamma_i}
  \gamma_1^\pm = \gamma(± X), \quad
  \gamma_2^\pm = \sigma_x^A \gamma(± Y) \sigma_x^A
\end{align}
where $\sigma_x^A$ is a Pauli matrix $\sigma_x$ applied on subsystem
$A$, and by $\gamma(X)$ we mean a private bit written in its $X$-form
\eqref{eq:pbit}.

States given by \eqref{eq:the-class} and \eqref{eq:the-class-gamma_i}
have the block matrix form \eqref{eq:the-class-blocks} given on figure
\ref{fig:the-class-blocks}.

\begin{definition}
  We define the class $\C$ as the class of states given by
  \eqref{eq:the-class} and \eqref{eq:the-class-gamma_i} with operators
  $X$ and $Y$ related by
  \begin{align}
    \label{eq:Y}
    Y =\frac{X^\Gamma}{\|X^\Gamma\|}
  \end{align}
  where superscript Γ denotes the partial transposition in Alice
  versus Bob cut; and satisfying the following conditions: the
  diagonal blocks of \eqref{eq:the-class-blocks}, i.e., operators
  $\sqrt{XX^\dagger}$, $\sqrt{X^\dagger X}$, $\sqrt{YY^\dagger}$,
  $\sqrt{Y^\dagger Y}$ are all PPT-invariant, i.e., must satisfy
  $A=A^\Gamma$.
\end{definition}
(The relation
\eqref{eq:Y} and PPT-invariance of the diagonal blocks are necessary to
obtain simple conditions for the state to be PPT, given in \mbox{Sec.}
\ref{sec:ppt}).

In particular, the PPT-invariance of the diagonal blocks holds for
\begin{align}
  \label{eq:X}
  X = \frac{1}{u} \sum_{i,j=0}^{d-1} u_{ij} |ij⟩⟨ji|
\end{align}
where $u_{ij}$ are elements of some unitary matrix on ${\cal C}^d$ and
\begin{align}
  \label{eq:u}
  u = \sum_{i,j=0}^{d-1} |u_{ij}|.
\end{align}
For the operator $X$ given by \eqref{eq:X} we have
\begin{align}
  \|X^\Gamma\| = \frac du, \qquad
  \frac{1}{\sqrt{d}} ≤ \|X^\Gamma\| ≤ 1
\end{align}
where the minimum is achieved for the unimodular unitary
\cite{lowdim-pptkey} and maximum for the identity matrix.

We will sometimes write $ρ_U$ to denote the subclass of the class
$\C$ with operator $X$ given by \eqref{eq:X} or to
stress using a concrete unitary in the definition of $X$, in particular, we
will consider the subclass $ρ_H$ where $u_{ij}$ are elements of the
Hadamard unitary matrix.

In case of $d=2$ we will also consider the subclass of the class
$\C$ with operators $X$ and $Y$ given by
\begin{align}
  \label{eq:spider-Y}
  Y = q \, Y_{U_1} + (1 - q) \, σ_x^{A'} Y_{U_2} σ_x^{A'}, \quad
  X = \frac{Y^Γ}{\|Y^Γ\|}
\end{align}
where
\begin{align}
  \label{eq:spider-Y_U}
  Y_U = \frac{1}{d} \sum_{i,j=0}^{d-1} u_{ij} |ii⟩⟨jj|.
\end{align}
Unitaries $U_1$ and $U_2$ must have the same global phase, i.e.,
$\alpha_1=\alpha_2$ in the parametrization of a single qubit unitary
given by \eqref{eq:qubit-unitary} in the appendix.  In particular, one
may take $U_1=U_2$.

We also use an alternative parametrization in terms of $p$, α, and β
given by
\begin{align}
  p &\equiv λ_1 + λ_2 \in [0, 1] \\
  \alpha &\equiv \frac{λ_1 - λ_2}{λ_1 + λ_2} \in [-1, 1] \\
  \beta &\equiv \frac{λ_3 - λ_4}{λ_3 + λ_4} \in [-1, 1].
\end{align}

On the other hand, the original parameters $λ_i$ can be expressed
using $p$, α, and β as follows:
\begin{align}
  λ_{1,2} &= \frac{1 ± α}{2} p \\
  λ_{3,4} &= \frac{1 ± β}{2} (1 - p).
\end{align}

Both parametrizations are directly related with the privacy-squeezed
version of the state given by $\C$ and
\eqref{eq:the-class-gamma_i}, and constructed according to the formula
\eqref{eq:ps-spider}:
\begin{align}
\label{eq:ps-class}
\sigma=\sum_i \lambda_i |ψ_i\>\<ψ_i|=\frac12
  \begin{bmatrix}
    p &  & & \alpha p \\
    \, & (1-p) & \beta (1-p) \\
    \, & \beta (1-p) & (1-p) \\
    \alpha p & & & p
  \end{bmatrix}
\end{align}
where the Bell states $ψ_i$ are given by
\begin{align}
  \label{eq:bell-states}
  |ψ_{1,2}⟩ &= \frac{1}{\sqrt2} (|00⟩±|11⟩) \nonumber \\
  |ψ_{3,4}⟩ &= \frac{1}{\sqrt2} (|01⟩±|10⟩).
\end{align}
Thus, $\lambda_i$ are the eigenvalues of the privacy-squeezed state,
$p$ reports the balance between correlations and anti-correlations,
while $\alpha$ and $\beta$ report how coherences are damped.

A subclass of the class $\C$ with $X$ defined by
\eqref{eq:X} has been considered in \cite{lowdim-pptkey}:
\begin{align}
  \label{eq:subclass}
  \tilde\varrho = λ_1 γ_1^+ + λ_3 γ_2^+.
\end{align}

The class $\C$ is much wider then \eqref{eq:subclass},
in particular, it contains key-distillable PPT states arbitrary close
to the separable states, but this comes with a price: we have to, in
general, use the recurrence preprocessing to obtain positive key rate for
$\C$ while for \eqref{eq:subclass} the sole
Devetak-Winter protocol is enough \cite{lowdim-pptkey}.

\subsection{Sufficient PPT conditions}
\label{sec:ppt}

For the states of the class $\C$ to be PPT (so that
maximal entanglement cannot be distilled from them) it is sufficient
to satisfy the following conditions
\begin{align}
  |λ_1 - λ_2| &≤ (1 - λ_1 - λ_2) \|X^Γ\|^{-1} \\
  \label{eq:ppt-cond-1b}
  |λ_3 - λ_4| &≤ (λ_1 + λ_2) \|X^Γ\|
\end{align}
or equivalently
\begin{align}
  \label{eq:ppt-cond-alpha}
  |\alpha| &≤ \min(1, \alpha_1) \\
  \label{eq:ppt-cond-beta}
  |\beta| &≤ \min(1, \alpha_1^{-1}) &
\end{align}
where
\begin{align}
  \label{eq:alpha1}
  \alpha_1 = \frac{1- p}{p} \|X^\Gamma\|^{-1}.
\end{align}
In particular, if $p=\tilde{λ}_1$ where $\tilde{λ}_1$ is given by
\eqref{eq:subclass-ppt-cond}, we have $\alpha_1=1$.  Moreover, if
$α=α_1β$ then ρ is a PPT-invariant state.

For the subclass \eqref{eq:subclass}, the above PPT conditions
collapse to a single PPT-invariant state, on the boundary of PPT
states, which satisfies
\begin{align}
  \label{eq:subclass-ppt-cond}
  λ_1 = \tilde{λ}_1 \equiv \frac{1}{1 + \|X^Γ\|}.
\end{align}

\subsection{Key distillability}
\label{sec:key}

We shall derive here a general sufficient condition for
key-distillability of the spider states with a Bell diagonal
privacy-squeezed state, which easily follows from combining
the privacy squeezing technique with the result of \cite{acin-2006-73} on
key distillation from two-qubit states.  It is enough for our
purposes, as states of our class are of that form.  (In \mbox{Sec.}
\ref{sec:general-key} we shall extend the key-distillability condition
to arbitrary states by exploiting twirling).

\begin{proposition}
  \label{prop:key}
  Let ρ be a state of the form
  \begin{align}
    \label{eq:th-key-rho}
    ρ =
    \begin{bmatrix}
      C & & & D \\
      & E & F \\
      & F^{†} & E' \\
      D^{†} & & & C'
    \end{bmatrix}
  \end{align}
  satisfying $\|C\|=\|C'\|$ and $\|E\|=\|E'\|$, i.e., ρ is a state
  having a Bell diagonal privacy-squeezed state.  If
  \begin{align}
    \label{eq:th-key-cond}
    \max(\|D\|,\|F\|) > \sqrt{\|C\| \|E\|}
  \end{align}
  then Alice and Bob can distill cryptographic key by first measuring
  the key part of many copies of the state ρ and than using the recurrence
  \cite{Maurer_key_agreement,BDSW1996} and the Devetak-Winter protocol
  \cite{DevetakWinter-hash}.
\end{proposition}

\noindent%
\begin{remark}
  Note that, interestingly, the condition \eqref{eq:th-key-cond} is
  equivalent to requiring that one of the matrices
  \begin{align}
    \begin{bmatrix}
      \|C\| & \|D\| \\
      \|D^\dagger\| & \|E\|
    \end{bmatrix},\quad
    \begin{bmatrix}
      \|C\| & \|F\| \\
      \|F^\dagger\| & \|E\|
    \end{bmatrix}\quad
  \end{align}
  is not a positive one.
\end{remark}

\begin{remark}
  Note that the right-hand side of \mbox{Eq.} \eqref{eq:th-key-cond}
  can also be written as $\frac12 \sqrt{p_e (1-p_e)}$ where $p_e$ is
  the probability of error (i.e. anticorrelation) when key part is
  measured in standard basis.
\end{remark}

\begin{proof}[Proof of the proposition \ref{prop:key}]
  We apply the privacy squeezing technique described in \mbox{Sec.}
  \ref{sec:prelim}, i.e., we show that the privacy-squeezed state of ρ
  is key-distillable by a protocol based on measuring the state
  locally in the standard basis and classical postprocessing.  This
  implies ρ is also key-distillable.

  The privacy-squeezed state is precisely of the form
  \eqref{eq:ps-spider} with $\|C\|=\|C'\|$ and $\|E\|=\|E'\|$, i.e.,
  it is a Bell diagonal state which can be written as
  \begin{align}
    \sigma = \frac12
    \begin{bmatrix}
      a &  & & d \\
      & e & f \\
      & f & e \\
      d & & & a
    \end{bmatrix}.
  \end{align}
  For such a state it was shown in \cite{acin-2006-73} that if
  $\max(|d|,|f|)>\sqrt{ae}$ then one can distill key by measuring the
  state locally in the standard basis, and processing the resulting
  classical data (actually, by using the recurrence followed by
  the Devetak-Winter protocol). This is precisely the type of protocols
  allowed by the privacy-squeezing technique described in \mbox{Sec.}
  \ref{sec:prelim}. In our case, the above conditions are simply the
  ones given in \eqref{eq:th-key-cond}.
\end{proof}

Due to the form \eqref{eq:ps-class} of the privacy-squeezed state of
the states from our class, we immediately obtain suitable conditions:
\begin{corollary}
  \label{cor:key}
  Let ρ be a state defined by formulas $\C$ and
  \eqref{eq:the-class-gamma_i} with arbitrary $X$ and $Y$ satisfying
  $\|X\|=\|Y\|=1$.  If
  \begin{align}
    \label{eq:key-cond-lambda}
    |λ_1 - λ_2| > \sqrt{(λ_1 + λ_2) (1 - λ_1 - λ_2)}
  \end{align}
  or equivalently if
  \begin{align}
    \label{eq:key-cond-alpha}
    |\alpha| > \sqrt{1 - p \over p}
  \end{align}
  then Alice and Bob can distill cryptographic key by first measuring
  the key part of many copies of the state ρ and than using the recurrence
  and the Devetak-Winter protocol.
\end{corollary}

Corollary \ref{cor:key} also holds if one uses $|λ_3 - λ_4|$ as the
left-hand side of
\eqref{eq:key-cond-lambda} or equivalently $|β|$ as the left-hand side of
\eqref{eq:key-cond-alpha}, however, in our paper, we do not use these
conditions.

\begin{observation}
  \label{obs:p-range-of-ppt-key}
  For a state of the class $\C$ to be both PPT and key
  distillable using corollary~\ref{cor:key} it must satisfy both
  \eqref{eq:ppt-cond-alpha} and \eqref{eq:key-cond-alpha}.  For a
  given value of the parameter $p$ there exist α satisfying both
  conditions iff $p \in (\frac12, p_{{\max}})$ where
  \begin{align}
     p_{{\max}} = \frac{1}{1 + \|X^Γ\|^2}.
  \end{align}
\end{observation}

\subsection{Tolerable white noise}

We say that δ is the \emph{tolerable noise} of a key distillation
protocol for a state ρ if for any $ε<δ$ the state $ρ_ε$ with ε of the
white noise admixtured
\begin{align}
  \varrho_ε = (1-ε) \varrho + ε \frac{I}{d^2}
\end{align}
remains key-distillable with that protocol.

Having $p>\frac12$, the tolerable noise of the Devetak-Winter protocol
with the recurrence preprocessing for the class $\C$ is
given by
\begin{align}
  \label{eq:tolerable-noise}
  \delta &= 1-{{1}\over{\sqrt{8(λ_1^2 + λ_2^2) - 4(λ_1 + λ_2) + 1}}} \\
  &= 1-{{1}\over{\sqrt{4\left(1 + α^2\right)\,p^2-4\,p+1}}}.
\end{align}

In particular for a key-distillable PPT state $\tilde{ρ}_H$ with
$λ_1=\tilde{λ}_1$ where $\tilde{λ}_1$ is given by
\eqref{eq:subclass-ppt-cond} the tolerable noise for the Devetak-Winter
protocol with the recurrence preprocessing \eqref{eq:tolerable-noise} is
approximately equal to 0.155 while for the sole Devetak-Winter
protocol it is approximately equal to 0.005, i.e., it is 31 times
smaller.  See figure~\ref{fig:noise}.

\begin{figure}
  \centering
  \includegraphics[width=8cm]{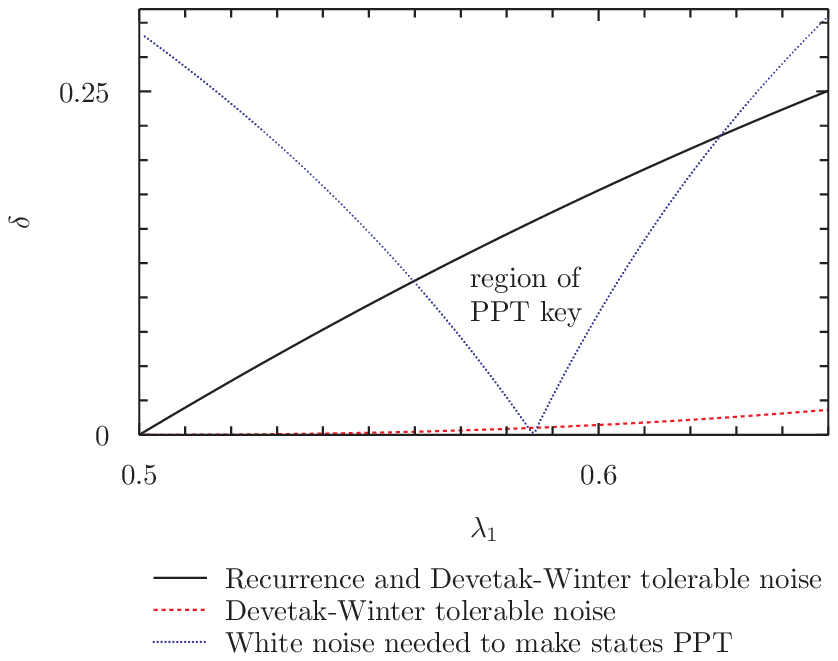}
  \caption{Comparison of $\tilde\varrho_H$ tolerable noise in case of
    using the Devetak-Winter protocol with and without the recurrence
    preprocessing.}
  \label{fig:noise}
\end{figure}

\subsection{Separability}

Given a state $ρ_U$ of the class $\C$ with $X$ given
by \eqref{eq:X} and $d=2$, i.e., ρ is a state of $4 ⊗ 4$ system, we may
try to decompose ρ into a mixture of four two qubit states.  The
particular decomposition, which we propose below, is possible if
\begin{align}
  \label{eq:sep-precond-lambda}
  |λ_3 - λ_4| ≤ (1 - λ_1 - λ_2) \|X^Γ\|
\end{align}
\begin{figure*}[t]
  \centering
  \begin{align}
    \label{eq:rho_ij}
    ρ_{ij} = \frac12
    \begin{bmatrix}
      λ_1 + λ_2 & & & (λ_1 - λ_2) e^{i ϕ_{ij}} \\
      & λ_3 + λ_4 & (λ_3 - λ_4) \|X^Γ\|^{-1} e^{i ϕ_{ij}} & \\
      & (λ_3 - λ_4) \|X^Γ\|^{-1} e^{-i ϕ_{ij}} & λ_3 + λ_4 & \\
      (λ_1 - λ_2)e^{-i ϕ_{ij}} &  &  & λ_1 + λ_2
    \end{bmatrix}
  \end{align}
  \caption{The form of two qubit Bell diagonal states from
    decomposition of a state $ρ_U$ with $d=2$.}
  \label{fig:rho_ij}
\end{figure*}%
or equivalently if
\begin{align}
  \label{eq:sep-precond-beta}
  |\beta| ≤ \|X^Γ\|.
\end{align}

All of the four two qubit states in our decomposition are Bell
diagonal states with the same set of eigenvalues.  Thus, the two qubit states
are separable (and, hence, ρ is separable) if all their eigenvalues are
less than or equal to $\frac12$ \footnote{This can be directly
  verified by positivity of partial transpose \cite{Peres96,sep1996}}.
For our decomposition this happens if, additionally to
\eqref{eq:sep-precond-lambda}, the following conditions are satisfied
\begin{align}
  λ_1 &≤ \frac12 \\
  λ_2 &≤ \frac12 \\
  \label{eq:sep-as-ppt-1}
  |λ_3 - λ_4| &≤ (λ_1 + λ_2) \|X^Γ\|
\end{align}
or equivalently, additionally to \eqref{eq:sep-precond-beta}, the
following conditions are satisfied
\begin{align}
  \label{eq:sep-alpha-cond}
  |α| & ≤ \frac{1-p}{p} \\
  \label{eq:sep-as-ppt-2}
  |β| & ≤ \frac{p}{1-p} \|X^Γ\|.
\end{align}
Note that conditions \eqref{eq:sep-as-ppt-1} and
\eqref{eq:sep-as-ppt-2} are identical to the PPT conditions for ρ
given by \eqref{eq:ppt-cond-1b} and \eqref{eq:ppt-cond-beta},
respectively.

The decomposition into the four two qubit states has the form
\begin{multline}
  \label{eq:4x2quit-decomposition}
  ρ_U = \frac{|u_{00}|}{u}ρ_{00}(|00⟩_{AA'}, |10⟩_{AA'}; |00⟩_{BB'}, |10⟩_{BB'}) \\
  + \frac{|u_{01}|}{u}ρ_{01}(|00⟩_{AA'}, |11⟩_{AA'}; |01⟩_{BB'}, |10⟩_{BB'}) \\
  + \frac{|u_{10}|}{u}ρ_{10}(|01⟩_{AA'}, |10⟩_{AA'}; |00⟩_{BB'}, |11⟩_{BB'}) \\
  + \frac{|u_{11}|}{u}ρ_{11}(|01⟩_{AA'}, |11⟩_{AA'}; |01⟩_{BB'}, |11⟩_{BB'})
\end{multline}
where $u_{ij}$ are the elements of the unitary matrix on ${\cal C}^2$
used to define operator $X$ in \eqref{eq:X}, $u$~is given by
\eqref{eq:u}, and $ρ_{ij}$ denote the two qubit states given by
\eqref{eq:rho_ij} on figure \ref{fig:rho_ij} where $ϕ_{ij}$ comes from
the polar decomposition of $u_{ij}$
\begin{align}
  u_{ij} = |u_{ij}|e^{i ϕ_{ij}}.
\end{align}

The local basis of Alice and Bob for each of the two qubit states are
given in \eqref{eq:4x2quit-decomposition} in parenthesis.

\subsection{PPT key arbitrary close to separability}

One can obtain key from some $4 ⊗ 4$ PPT states lying arbitrary close
to the set of separable states.  That is, one can easily select a
single parameter subclass of the class $\C$ satisfying
PPT conditions and approaching some separable state with $p=\frac12$
such that for any other state in this class, no matter how close to
the separable state, the key condition \eqref{eq:key-cond-alpha} is
satisfied.  Note that if we chose a separable state with $p \ne
\frac12$ as the final state the key condition would be violated before
reaching that final state; thus, we would not approach with the
key-distillable states arbitrary close to the set of separable states.

\begin{figure}
  \centering
  \includegraphics{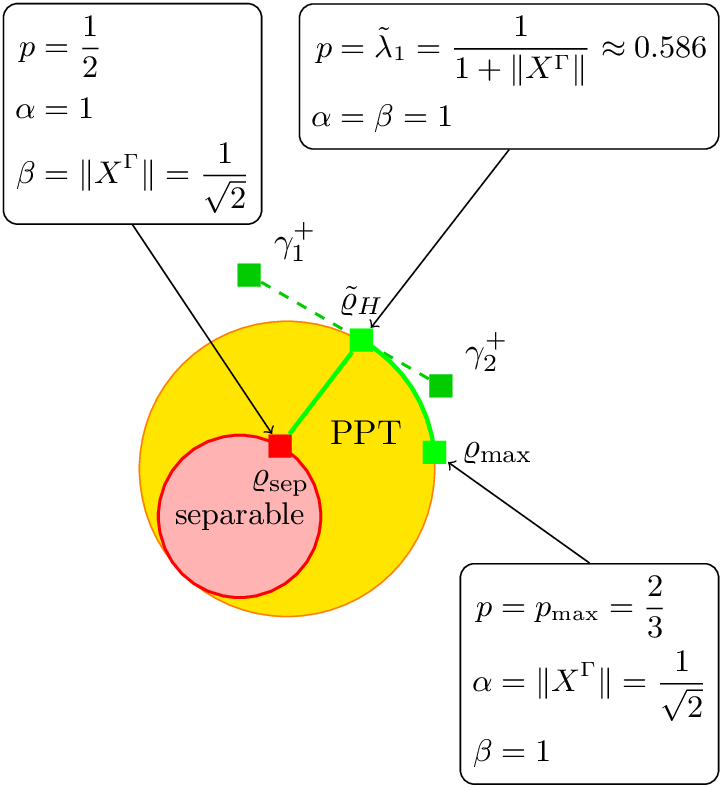}
  \caption{A class of key-distillable PPT entangled states: (a) the
    solid line from $\tilde{ρ}$ on the boundary of the PPT entangled
    states (inclusive) to the boundary of the set of separable states,
    arbitrary close to $ρ_{\mathrm{sep}}$; (b) the arc of
    PPT-invariant states starting in $\tilde{ρ}$ and approaching
    arbitrary close to $\tilde{ρ}_{{\max}}$.}
  \label{fig:key-up-to-sep}
\end{figure}

Such a class of states, a subclass of $ρ_H$, is illustrated in figure
\ref{fig:key-up-to-sep}.  The dashed line represents the subclass
$\tilde\varrho_H$, given by \eqref{eq:subclass}, a mixture of two
pbits ($γ_1^+$ and $γ_2^+$) which in alternate parametrization is
equivalent to $p\in[0,1]$ and $\alpha=\beta=1$.  As shown in
\cite{lowdim-pptkey}, this class contains exactly one (boundary) PPT
entangled state obtained by setting $p=\tilde{λ}_1$ where
$\tilde{λ}_1$ is given by \eqref{eq:subclass-ppt-cond}, otherwise the
states are NPT.

The solid line represents a class of PPT key-distillable states
obtained by setting $p\in(\frac12,p_{{\max}})$, $α=\min(1,α_1)$, and
$β=\min(1, α_1^{-1})$, where $p_{{\max}}=(1+\|X^Γ\|^2)^{-1}=\frac23$,
see observation \ref{obs:p-range-of-ppt-key}, while $α_1$ is given by
\eqref{eq:alpha1}, i.e., $α_1=\frac{1-p}{p}\sqrt{2}$ in the considered
case.  In the range $p\in(\frac12,\tilde{λ}_1]$ the class
is represented as a straight line from the PPT state of the previous
class $\tilde{ρ}_H$ on one end ($p=\tilde{λ}_1$) and approaches
arbitrary close to the separable state $ρ_{\mathrm{sep}}$ ($p=\frac12$)
on the other end.  In the range $p\in[\tilde{λ}_1, p_{{\max}})$ the
states are PPT-invariant and lie on the boundary of PPT entangled
states, they are represented as an arc from the PPT state of the
previous class $\tilde{ρ}_H$ on one end ($p=\tilde{λ}_1$) and approach
arbitrary close to the state $ρ_{{\max}}$ ($p=p_{{\max}}$) on the other
end.  In the range $p\in(\tilde{λ}_1, p_{{\max}})$ one could take
$α<α_1$, such that the key condition \eqref{eq:key-cond-alpha} is
still satisfied, to enter inside the class of PPT states.

\section{States $\varrho_H$ as mixtures of Bell states with `flags'}
\label{sec:exp}

States of the class $\varrho_H$ are separable in the $AB:A'B'$ cut,
i.e., subsystems $AB$ and $A'B'$ of $\varrho_H$ are are only
classically correlated.  A state from $ρ_H$ can be decomposed into a
mixture of four states.  Each of the four states has a Bell state
$ψ_i$ on the subsystem $AB$ and some corresponding state on $A'B'$.

One can select parameters $p ∈ [0,1]$, $α ∈ [-1,1]$, and $β ∈ [-1,1]$
satisfying both the PPT conditions \eqref{eq:ppt-cond-alpha} and
\eqref{eq:ppt-cond-beta} and the key condition
\eqref{eq:key-cond-alpha}, and prepare a corresponding PPT
key-distillable state from the class $ρ_H$ which has the form
\begin{align}
  \varrho_H = \sum_{i=1}^4 q_i \, |ψ_i⟩⟨ψ_i|_{AB} \otimes \varrho_{A'B'}^{(i)}
\end{align}
where the Bell states $ψ_i$ are given by \eqref{eq:bell-states} and the
correlated states are the following:
\begin{align}
  \varrho^{(1)}
  &= \alpha \frac12 (P_{00} + P_{\psi_3}) + (1-\alpha) \frac{I}{4} \\
  \varrho^{(2)}
  &= \alpha \frac12 (P_{11} + P_{\psi_4}) + (1-\alpha) \frac{I}{4} \\
  \varrho^{(3,4)} &= \beta P_{\chi_\pm} + (1-\beta) \frac12 (P_{00} + P_{11})
\end{align}
where $P_{ψ}$ denotes the projector onto a pure state ψ and
\begin{align}
  \chi_\pm &= \frac{1}{\sqrt{2 \pm \sqrt2}} (|00\> \pm |\psi_1\>) \\
  q_1 &= q_2 = \frac{p}{2} \\
  q_3 &= q_4 = \frac{1 - p}{2}.
\end{align}

\section{Maximizing von Neumann Entropy}
\label{sec:entropy}

In this section, we find $4\otimes4$ key-distillable PPT states with
a quite high von Neumann entropy for two subclasses of the class
$\C$ and summarize the results in a table.

\subsection{For states of the class $ρ_U$}
\label{sec:sup-S-rho_U}

Here, we find the supremum of the von Neumann entropy of the subclass
$ρ_U$ of the class $\C$ with $X$ given by \eqref{eq:X}
consisting of states that are both PPT and key-distillable by
corollary~\ref{cor:key}.  Let as denote this set of states as ${\cal
  PK}_d$, subscripted with the dimension of the unitary used to define
operator $X$.

As $\varrho$ is a mixture of four orthogonal private bits its von
Neumann entropy is given by
\begin{multline}
  \label{eq:S}
  S(\varrho_U) = H(p)
  + p \left(H\left(\frac{1-\alpha}{2}\right)
    + S(\sqrt{X^\dagger X})\right) \\
  + (1 - p) \left(H\left(\frac{1-\beta}{2}\right)
    + S(\sqrt{Y^\dagger Y})\right)
\end{multline}
where
\begin{align}
  \label{eq:S-XX}
  S(\sqrt{X^\dagger X}) &≤ 2\log_2 d \\
  S(\sqrt{Y^\dagger Y}) &= \log_2 d
\end{align}
and the maximal value in \eqref{eq:S-XX} is achieved if the unitary
used to define $X$ in \eqref{eq:X} is unimodular.  A unimodular
unitary also maximizes the allowed range of $p$ given by observation
\ref{obs:p-range-of-ppt-key}, as it achieves minimum of $\|X^Γ\|$.
Hence, to maximize the entropy, it is enough to consider a unimodular
unitary.  The supremum is achieved for a state with $p=p_{{\max}}$,
$β=0$, and $\alpha=\sqrt{\frac{1-p}{p}}$ (which no longer satisfies
our key-distillability condition) thus
\begin{multline}
  \sup_{\varrho_U\in {\cal PK}_d} S(\varrho_U) =
  \sup_{p\in(\frac12, p_{\max})}
  \Bigg(
  (1+p)\log_2 d + (1 - p) \\
  + H(p) + p H\left(\textstyle\frac{1-\sqrt{\frac{1 - p}{p}}}{2}\right)
 \Bigg)
\end{multline}
where $p_{\max}= (1 + \|X^Γ\|^2)^{-1}$ comes from observation
\ref{obs:p-range-of-ppt-key}.

In particular, for $d=2$, i.e., ρ being $4 ⊗ 4$ states, the supremum
is achieved for state having $p=p_{\max}=2/3$ which gives
\begin{align}
  \sup_{\varrho_U\in {\cal PK}_2} S(\varrho_U) \approx 3.319.
\end{align}
The supremum corresponds to a state $ρ_{{\max}}$ on figure
\ref{fig:key-up-to-sep} but with $β=0$.

\subsection{For states of a class larger than $ρ_U$}
\label{sec:max-S-spider-Y}

For the subclass ρ of the class $\C$ with $d=2$ and $X$
and $Y$ given by \eqref{eq:spider-Y}, we are able to obtain
\begin{align}
  S(ρ) \approx 3.524
\end{align}
for $U_1=U_2=H$, $q \approx 0.683$, $β=0$ and α, $p$ taken as in the
previous subsection.  It seems to be the supremum of the von Neumann
entropy for this selection of operators $X$ and $Y$.

\subsection{Summary}

Here, we summarize the results of maximizing von Neumann entropy of
$4\otimes4$ key-distillable PPT states in the following table:
\medskip

\noindent
\begin{tabular}{lp{7.5cm}}
  \toprule
  $S(ρ)$ & $ρ$ satisfying PPT and key conditions \\
  \midrule
  2.564 & class $\tilde{ρ}$ from \cite{lowdim-pptkey}
  with $p=\tilde{λ}_1$, the maximum is achieved for $U=H$ \\
  3.319 & class $ρ_U$, the supremum is described
  in \mbox{Sec.} \ref{sec:sup-S-rho_U} \\
  3.524 & class $\C$ with $Y$ given by \eqref{eq:spider-Y},
  a supposed supremum is described in \mbox{Sec.} \ref{sec:max-S-spider-Y} \\
  \bottomrule
\end{tabular}

\section{Distillability via erasure channel}
\label{sec:erasure}

In \cite{SmithYard-2008}, it was shown that two zero capacity
channels, if combined together, can have nonzero capacity. One of the
channels was related (through so called Choi-Jamio\l{}kowski (CJ)
isomorphism) to a bound entangled but key-distillable state, while the
other was a so called symmetrically extendable channel.  In
particular, they considered an example, where the first channel had $4
⊗ 4$ CJ state from the class \eqref{eq:subclass} while the second one
was the 50\%-erasure channel.  In
\cite{Jonathan-two-wrongs-make-right} a simpler scheme was proposed,
which also allows to observe this curious phenomenon.

The second approach amounts to sending a subsystem $A'$ of a state
defined on systems $ABA'B'$ through the 50\%-erasure channel and
checking the coherent information of the resulting state. If it is
positive one concludes that the capacity of combined channel is also
positive.  Here, we shall use this approach to see how the presence of
coherence $\beta$ influence the phenomenon.

Coherent information after sending the $A'$ subsystem through the
50\%-erasure channel is given by
\begin{align}
  I_{\mathrm{coh}} = \frac12 (S_{A'BB'} - S) + \frac12 (S_{BB'} - S_{ABB'})
\end{align}
where $S$, $S_{A'BB'}$, and $S_{BB'}$ are given by \eqref{eq:S},
\eqref{eq:S_A'BB'}, and \eqref{eq:S_BB'}, respectively.

For a PPT state $\tilde{ρ}$ given by \eqref{eq:subclass} with $X$
given by \eqref{eq:X} and based on unimodular unitary and
$λ_1=\tilde{λ}_1$, where $\tilde{λ}_1$ is given by
\eqref{eq:subclass-ppt-cond}, the coherent information is positive
starting from $d=11$.  For a similar state of our class with
$p=\tilde{λ}_1$, $α=1$ and $β=0$ the coherent information is positive
starting from $d=22$.

Formulas for $S_{A'BB'}$ and $S_{BB'}$ are as follows:
\begin{multline}
  \label{eq:S_A'BB'}
  S(\varrho_{A'BB'}) = 1
  + \frac12 S\left( p \sqrt{X X^\dagger} + (1 - p) \sqrt{Y^\dagger Y}\right) \\
  + \frac12 S\left( p \sqrt{X^\dagger X} + (1 - p) \sqrt{Y Y^\dagger}\right)
\end{multline}
\begin{multline}
  \label{eq:S_BB'}
  S(\varrho_{BB'}) = 1
  + \frac12 S_B\left( p \sqrt{X X^\dagger} + (1 - p) \sqrt{Y^\dagger Y}\right) \\
  + \frac12 S_B\left( p \sqrt{X^\dagger X} + (1 - p) \sqrt{Y Y^\dagger}\right).
\end{multline}

\section{Condition for drawing secure key from general states}
\label{sec:general-key}

From \mbox{Sec.} \ref{sec:key}, we have a sufficient condition for
drawing key in terms of norms of the nonzero blocks from states having a
Bell diagonal privacy-squeezed state.  In this section, we
generalize that condition to the case of an arbitrary state.

Let us define two twirling operations (\mbox{cf.} \cite{BDSW1996})
\begin{align}
  Λ_{XX} &= \frac12 ( \hat I ⊗ \hat I + \hat X ⊗ \hat X ) \\
  Λ_{ZZ} &= \frac12 ( \hat I ⊗ \hat I + \hat Z ⊗ \hat Z )
\end{align}
and one twirling with flags
\begin{align}
  Λ_{XX}'(ρ) &= \frac12 ( ρ ⊗ |0⟩⟨0| + \hat X ⊗ \hat X (ρ) ⊗ |1⟩⟨1| )
\end{align}
where $\hat U ρ= U ρ U^†$, $X$ and $Z$ are Pauli matrices.

Now, we give a sufficient condition to obtain key from a general
state.

\begin{proposition}
  \label{prop:generic-key}
  For an arbitrary state
  \begin{align}
    \varrho =
    \begin{bmatrix}
      A & B & C & D \\
      B^{†} & E & F & G \\
      C^{†} & F^{†} & H & I \\
      D^{†} & G^{†} & I^{†} & J
    \end{bmatrix}
  \end{align}
  if
  \begin{align}
    \label{eq:th-general}
    \max(\|D\|,\|F\|) > \frac12 \sqrt{(\|A\|+\|J\|)(\|E\| + \|H\|)}
  \end{align}
  then Alice and Bob can distill cryptographic key by first applying
  twirling $Λ_{XX}' \circ Λ_{ZZ}$ to the key part and measuring the
  key part of many
  copies of the state ρ and than using the recurrence and the Devetak-Winter
  protocol.
\end{proposition}

\begin{remark}
  Note that the right-hand side of \mbox{Eq.} \eqref{eq:th-general}
  can also be written as $\frac12 \sqrt{p_e (1-p_e)}$ where $p_e$ is
  the probability of error (i.e. anticorrelation) when key part is
  measured in standard basis.
\end{remark}

\begin{proof}[Proof of the proposition \ref{prop:generic-key}]
  Alice and Bob first apply twirling $Λ_{XX}' \circ Λ_{ZZ}$ (an LOCC
  operation) to the key part and obtain the following state
  \begin{multline}
    Λ'_{XX} \circ Λ_{ZZ} (ρ)\\
    =\begin{bmatrix}
      A \oplus J & & & D\oplus D^† \\
      & E  \oplus H & F  \oplus F^† \\
      & F  \oplus F^† &  E  \oplus H \\
      D  \oplus D^† & & & A  \oplus J
    \end{bmatrix}.
  \end{multline}
  This state is now of the spider form and, thanks to flags, we have
  direct sums within the blocks. Now, the privacy-squeezed state has
  the following Bell diagonal form
  \begin{multline}
    \sigma = \\
    =\small\begin{bmatrix}
      \|A\|  + \|J\| & & & \|D\|  + \|D^†\| \\
      & \|E\|  + \|H\| & \|F\|  + \|F^†\| \\
      & \|F\|  + \|F^†\| &  \|E\|  + \|H\| \\
      \|D\|  + \|D^†\| & & & \|A\|  + \|J\|
    \end{bmatrix}.
  \end{multline}
  Then the proof follows from proposition~\ref{prop:key}.
\end{proof}

Note that in the proof above we use $Λ_{XX}'$, a twirling with flags.
If $Λ_{XX}$, a twirling without flags, were used instead we would have
to replace $\|D\|$ with $\|D + D^†\|$ in \eqref{eq:th-general}
(analogously for $\|F\|$) which can be much smaller than $\|D\|$, and
even equal to zero in the extreme case of antihermitian $D$, i.e.,
$D^†=-D$, so in this case no key can be distilled from $Λ_{XX}(ρ)$
even if ρ is a private state, i.e., $ρ=γ(D)$.

Note also, that in the proof, we have first applied twirling with
flags to the original state, and then the privacy-squeezing
operation. Actually, the same state would be obtained if we first
apply the privacy squeezing and then apply (standard) twirling. This
is illustrated by the following diagram
\begin{align}
  \begin{CD}
    ρ @>{Λ'_{XX} \circ Λ_{ZZ}}>> ρ' \\
    @V{P_{sq}}VV @VV{P_{sq}}V \\
    σ @>{Λ_{XX} \circ Λ_{ZZ}}>> σ'
  \end{CD}
\end{align}
where $P_{sq}$ stands for the privacy squeezing.  As explained above, this
diagram would not commute if we used solely twirling without flags.
Thus, to seek for key-distillable states, one can go the alternative
route, i.e., first compute the privacy-squeezed state, and then, by
twirling, obtain a Bell diagonal state.  Now, if $Λ_{XX} \circ
Λ_{ZZ}(σ)$ satisfies necessary security condition for realistic QKD on
a Pauli channel from \cite{acin-2006-73}, i.e., its eigenvalues $λ_i$
satisfy \eqref{eq:key-cond-lambda}, then ρ is key-distillable using
proposition~\ref{prop:generic-key}.

\section{Appendix}

The parametrization of a single qubit unitary \cite{Nielsen-Chuang}:
\begin{align}
  \label{eq:qubit-unitary}
  U &= e^{i\alpha}
  \begin{bmatrix}
    e^{i\*\left(-\frac{\beta}{2}-\frac{\delta}{2}\right)}
    \*\cos\left(\frac{\gamma}{2}\right) &
    -e^{i\*\left(-\frac{\beta}{2} + \frac{\delta}{2}\right)}
    \*\sin \left(\frac{\gamma}{2}\right)\\[0.2ex]
    e^{i \*\left(\frac{\beta}{2}-\frac{\delta}{2}\right)}
    \*\sin \left(\frac{\gamma}{2} \right) &
    e^{i\*\left(\frac{\beta}{2}+\frac{\delta}{2}\right)}
    \*\cos \left(\frac{\gamma}{2}\right)
  \end{bmatrix}.
\end{align}

\section{Acknowledgment}

The work is supported by Polish Ministry of Science and Higher
Education grant \mbox{no.} 3582/B/H03/2009/36 and by the European
Commission through the Integrated Project FET/QIPC QESSENCE.  This
work was done in National Quantum Information Centre of Gda\'nsk.

\vfill

\bibliographystyle{apsrev4-1long}
\bibliography{lupan,rmp12-hugekey}

\end{document}